\documentclass[conference]{IEEEtran}
\IEEEoverridecommandlockouts
\usepackage{cite}
\usepackage{amsmath,amssymb,amsfonts}
\usepackage{amsthm}
\usepackage{algorithmic}
\usepackage{graphicx}
\usepackage{textcomp}
\usepackage{xcolor}
\def\BibTeX{{\rm B\kern-.05em{\sc i\kern-.025em b}\kern-.08em
    T\kern-.1667em\lower.7ex\hbox{E}\kern-.125emX}}

\newtheorem{theorem}{Theorem}

\begin{document}

\title{Covert and Reliable Short-Packet Communications against A Proactive Warder}

\author{\IEEEauthorblockN{Manlin~Wang\textsuperscript{\dag},~Yao~Yao\textsuperscript{\dag},~Bin~Xia\textsuperscript{\dag},~Zhiyong~Chen\textsuperscript{\dag}~and Jiangzhou~Wang\textsuperscript{\ddag}}
\IEEEauthorblockA{\textsuperscript{\dag}\textit{Department of Electronic Engineering, Shanghai Jiao Tong University, Shanghai, China} \\
\textsuperscript{\ddag}\textit{School of Engineering, University of Kent, Canterbury, U.K.}\\
Email: \textsuperscript{\dag}\{wangmanlin, sandyyao, bxia, zhiyongchen\}@sjtu.edu.cn,\textsuperscript{\ddag}j.z.wang@kent.ac.uk}

\thanks{This work is supported in part by the National Natural Science Foundation of China under Grant 62271309, and Shanghai Municipal Science and Technology Major Project under grant 2021SHZDZX0102.}

}

\maketitle

\begin{abstract}
Wireless short-packet communications pose challenges to the security and reliability of the transmission. Besides, the proactive warder compounds these challenges, who detects and interferes with the potential transmission. An extra jamming channel is introduced by the proactive warder compared with the passive one, resulting in the inapplicability of analytical methods and results in exsiting works. Thus, effective system design schemes are required for short-packet communications against the proactive warder. To address this issue, we consider the analysis and design of covert and reliable transmissions for above systems. Specifically, to investigate the reliable and covert performance of the system, detection error probability at the warder and decoding error probability at the receiver are derived, which is affected by both the transmit power and the jamming power. Furthermore, to maximize the effective throughput, an optimization framework is proposed under reliability and covertness constraints. Numerical results verify the accuracy of analytical results and the feasibility of the optimization framework. It is shown that the tradeoff between transmission reliability and covertness is changed by the proactive warder compared with the passive one. Besides, it is shown that longer blocklength is always beneficial to improve the throughput for systems with optimized transmission rates. But when transmission rates are fixed, the blocklength should be carefully designed since the maximum one is not optimal in this case.
\end{abstract}


\section{Introduction}
Time-sensitive and mission-critical Internet of Things (IoT) applications have aroused great attention in the fifth-generation mobile communications systems \cite{Intro_short}. The use of short packets meets the stringent low latency requirements, but a severe loss in coding gain exits with short packets, posing challenges to transmission reliability. Besides, massive confident messages are transmitted in wireless channels in IoT scenarios, which poses unprecedented challenges to transmission security. The exposure of transmission behaviors may bring unpredictable risks and losses in these scenarios. Notably, covert communication offers an solution for this issue, which prevents the transmission behaviors from being detected\cite{Covert}.

The fundamental work for covert communication \cite{Covert} demonstrated that $\mathcal{O}(\sqrt{n})$ bits of information can be transmitted reliably and deniably over $n$ channel use. In addition, considering the covert communication with short packets, \cite{Gaussian} investigated the effective throughput of the system in additive white Gaussian noise (AWGN) channels. Similarly, \cite{AWGN3} considered the achievability bounds on the maximal channel coding rate at a given blocklength and error probability over AWGN channels. In addition, \cite{Fading3,Nmin} investigated the throughput over quasi-static fading channels, revealing the fundamental difference in the design between the case of quasi-static fading channel and that of AWGN channel. More complex scenarios with multiple warders, multi-antenna sources and unmanned aerial vehicle aided networks were considered in \cite{ECT2,Fading4,UAV}.

The warder in the aforementioned works is passive, who aims to detect the transmission behaviours while not degrading the quality of communication channels. Different from the passive warder, the proactive one behaves more dangerously. This is because the proactive warder can not only detect the wireless transmission, but also emit noise to interfere with the potential transmission simultaneously \cite{Active1}. In \cite{Active2}, a proactive warder was considered in the relay networks, where the behaviors of the transmitter and the warder were modeled as the non-cooperative game. In addition, \cite{Active3} investigated the issues of power control in the device-to-device covert communication networks consisting of a proactive warder.

However, the analysis and corresponding system designs about the proactive warder in \cite{Active1,Active2,Active3} were based on an infinite blocklength assumption, which is no longer suitable for short-packet transmission. Besides, the results from the system with passive warders \cite{Fading3,Nmin,ECT2,Fading4,UAV} can not be directly applied to the system with the proactive warder, since another jamming link exits between the warder and the destination in addition to communication and detection links. This compounds the challenges of both reliability and covertness in short-packet communications. Therefore, effective short-packet transmission design schemes to provide both reliability and covertness guarantees are still open issues. 

To address this issue, in this paper, we consider the analysis and design of reliable and covert transmissions against a proactive warder. Specifically, to guarantee the system covertness requirement, the average detection error probability at the warder is derived. Besides, to facilitate system analysis and optimization, concise approximation expression is also proposed, which is tighter than the widely used Kullback–Leibler (KL) divergence approximation in existing works on covert communications. To guarantee the reliability requirement, the average decoding error probability at the receiver is derived. Furthermore, an optimization problem is formulated and an optimization framework is proposed to maximize the effective throughput of the system with reliability and covertness constraints by jointly designing the transmit power, transmission rate and blocklength. Numerical simulations verify the tightness of the proposed approximations and the feasibility of the proposed optimization framework for the system.

\emph{Notation}: $\left|  \cdot  \right|$ denote the absolute value operator. $\mathcal{C} \mathcal{N}(0, {\sigma^2})$ denotes the complex Gaussian distribution with zero mean and variance ${\sigma^2}$. ${\Pr}(\cdot)$ denotes the probability of an event. $\mathcal{Q}(x)=\int_x^{\infty} \frac{1}{\sqrt{2 \pi}} \exp \left(-t^2 / 2\right) d t$ denotes the Q-function. $\Gamma \left( n \right) = \left( {n - 1} \right)!$ denotes the Gamma function, and $\gamma(n, x)=\int_0^x e^{-t} t^{n-1} d t$ denotes the lower incomplete Gamma function. $\psi(x)=\frac{{d\ln \left( {\Gamma \left( x \right)} \right)}}{{dx}}$ denotes the digamma function while $\psi^{(n)}(x)$ denotes its $n$-th derivative. ${\rm{E_1}}(x) = \int_x^\infty  {\frac{{{e^{ - t}}}}{t}dt} $ denotes the exponential integral function.

\section{System Model}
\subsection{Signal and Channel Models}
\begin{figure}
	\centering
	\includegraphics[width=0.65\linewidth]{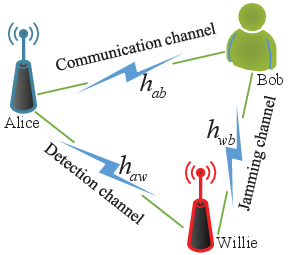}
	\caption{Covert and reliable communication system against a proactive warder.}
	\label{fig_sim}
\end{figure}
As shown in Fig. 1, a covert wireless communication scenario is considered, where the transmitter (Alice) desires to deliver messages to the receiver (Bob) while keeping a full-duplex warder (Willie) unaware of the transmission. Willie operates in full-duplex receiving signals from Alice and transmitting jamming signals to Bob simultaneously. Alice and Bob are assumed to be equipped with a single antenna, while Willie is assumed to be equipped with two antennas to support full-duplex functionality (detecting and jamming) \cite{Active3}.

In one transmission round, Alice transmits $n$ covert signals $x_a\left[ i \right],i \in \left\{ {1, \cdots,n} \right\}$ to Bob, while Willie sends $n$ jamming signals $x_w[i],i \in \left\{ {1, \cdots,n} \right\}$. Besides, Willie collects $n$ received signals to detect whether or not Alice has transmitted signals. The transmit power of Alice is denoted as $P_a$ and $x_a\left[ i \right] \sim \mathcal{CN}\left( {0,P_a} \right)$ \cite{Gaussian}. Similarly, the jamming power of Willie is denoted as $P_w$ and $x_w\left[ i \right] \sim \mathcal{CN}\left( {0,P_w} \right)$. We denote the AWGN at Bob and Willie as ${n_b}\left[ i \right] \sim \mathcal{CN}\left( {0,\sigma _b^2} \right)$ and ${n_w}\left[ i \right] \sim \mathcal{CN}\left( {0,\sigma_w^2} \right)$, where $\sigma _b^2$ and $\sigma_w^2$ are the noise variances at Bob and Willie, respectively.

The wireless channels from Alice to Bob (communication channel, $h_{ab}$), Alice to Willie (detection channel, $h_{aw}$) and Willie to Bob (jamming channel, $h_{wb}$) are subject to the quasi-static Rayleigh fading \cite{DF}. Specifically, $h_{ab} \sim \mathcal{CN}(0,\lambda_{ab})$, $h_{aw} \sim \mathcal{CN}(0,\lambda_{aw})$ and $h_{wb} \sim \mathcal{CN}(0,\lambda_{wb})$. The channel coefficients remain constant during one transmission round, and are independently and identically distributed (i.i.d.) among different rounds. The instantaneous channel state information (CSI) $h_{aw}$ is unavailable for Alice since Willie does not cooperate with Alice as an adversarial node while the statistical CSI is able to be estimated through the jamming signal \cite{statistical_CSI}. Besides, the instantaneous CSI is available for Willie from a worst case perspective for covert communication.

\subsection{Binary Hypothesis Testing at Willie}
In order to detect the presence of covert communications, Willie must distinguish between the following two hypotheses in each transmission round
\begin{equation}
	y_w[i]= \begin{cases}  \sqrt\varphi  {x_w}[i]+{n_w}[i], & \mathcal{H}_{0} \\{h}_{aw}{x_a}[i] +  \sqrt\varphi  {x_w}[i]+ {n_w}[i], & \mathcal{H}_{1} \end{cases}
\end{equation}
where $\mathcal{H}_{0}$ denotes the null hypothesis where Alice has not
transmitted, $\mathcal{H}_{1}$ denotes the alternative hypothesis
where Alice has transmitted. $y_w[i]$ is the received signal at Willie, and $\varphi \in \left[ {0,1} \right]$ is the self-interference cancellation coefficient\cite{Active2,Active3}.

With a radiometer \cite{Gaussian}, Willie makes a binary decision as
\begin{equation}
	T=\frac{1}{n} \sum_{i=1}^n\left|{y}_w[i]\right|^2 \underset{\mathcal{D}_0}{\stackrel{\mathcal{D}_1}{\gtrless}} \tau,
\end{equation}
where $T$ is the average power of each received signal at Willie, $\tau$ denotes the detection threshold, $\mathcal{D}_0$ and $\mathcal{D}_1$ denote the binary decisions that infer whether Alice transmits or not.

Suppose there is no prior knowledge for Willie about when Alice will transmit, the priori probability of either hypothesis is equal. Mathematically,
the detection error probability $\xi$ at Willie is defined as follows \cite{Nmin,Active2,Active3,DF}
\begin{equation}\label{Pd}
	\begin{aligned}
	\xi \left( \tau  \right)&={\Pr}\left(\mathcal{D}_1 \mid \mathcal{H}_0\right)+{\Pr}\left(\mathcal{D}_0 \mid \mathcal{H}_1\right) \\
	&= \Pr \left( {T > \tau |{H_0}} \right) + \Pr \left( {T < \tau |{H_1}} \right),
	\end{aligned}
\end{equation}
where $ {\Pr}\left(\mathcal{D}_1 \mid \mathcal{H}_0\right)$ denotes the false alarm probability, and $ {\Pr}\left(\mathcal{D}_0 \mid \mathcal{H}_1\right)$ denotes the missed detection probability. In covert communications, Willie’s
ultimate goal is to detect the presence of Alice’s transmission
with the minimum detection error probability $\xi^*$, which is
achieved by using the optimal detection threshold $\tau^*$ that minimizes $\xi$.

\subsection{Effective Throughput with Finite Blocklength}
When Alice transmits, the received signal at Bob can be expressed as
\begin{equation}\label{y_b}
	y_b[i]=h_{ab}x_a[i]+{h_{wb}}{x_w}[i]+n_b[i].
\end{equation}

Based on the received signal (\ref{y_b}), Bob can decode the messages. The decoding error cannot be ignored in short-packet communications, which is given by \cite{ECT2}
\begin{equation}\label{Pout}
	\delta=\mathcal{Q}\left(\frac{\ln 2 \sqrt{n}\left(\log _2\left(1+\gamma_b\right)-R\right)}{\sqrt{1-\left(\gamma_b+1\right)^{-2}}}\right),
\end{equation}
where ${\gamma _b} = {{{P_a}{{\left| {{h_{ab}}} \right|}^2}} \mathord{\left/
		{\vphantom {{{P_a}{{\left| {{h_{ab}}} \right|}^2}} {\left( {{P_w}{{\left| {{h_{wb}}} \right|}^2} + \sigma _b^2} \right)}}} \right.
		\kern-\nulldelimiterspace} {\left( {{P_w}{{\left| {{h_{wb}}} \right|}^2} + \sigma _b^2} \right)}}$ denotes the received signal to noise ratio (SNR) at Bob, and $R$ is the transmission rate measured by bits per channel use (bpcu).

Since the decoding error probability (\ref{Pout}) is affected by fading channels $h_{ab}$ and $h_{wb}$, the average decoding error probability $\overline \delta$ is adopted to evaluate the reliability performance. And the effective throughput of the system is given by\cite{ECT2}
\begin{equation}\label{obj}
	\eta  = nR\left( {1 - \overline \delta  } \right),
\end{equation}
which quantifies the expected number of information bits that can be reliably transmitted from Alice to Bob.

\section{Covertness Performance Analysis}
In this section, to analyze the covertness performance of the system, the average detection error probability is derived.

With detection threshold $\tau$, the detection error probability is expressed as \cite{Gaussian}
\begin{equation}\label{ori}
	\xi \left( \tau  \right) = 1 - \frac{{\gamma \left( {n,\frac{{n\tau }}{{{\sigma ^2}}}} \right)}}{{\Gamma \left( n \right)}} + \frac{{\gamma \left( {n,\frac{{n\tau }}{{{\sigma ^2} + {P_a}{{\left| {{h_{aw}}} \right|}^2}}}} \right)}}{{\Gamma \left( n \right)}},
\end{equation}
where ${\sigma ^2} = \varphi {P_w} + \sigma _w^2$ for expression simplification.

Since Willie knows $h_{aw}$ in each round, Willie can adjust the optimal threshold $\tau^*$ to minimize the detection error probability for each round, which is given by \cite{Gaussian}
\begin{equation}\label{opt_threshold}
	{\tau ^*} = \frac{{{\sigma ^2}\left( {{\sigma ^2} + {P_a}{{\left| {{h_{aw}}} \right|}^2}} \right)}}{{{P_a}{{\left| {{h_{aw}}} \right|}^2}}}\ln \left( {\frac{{{\sigma ^2} + {P_a}{{\left| {{h_{aw}}} \right|}^2}}}{{{\sigma ^2}}}} \right).
\end{equation}

Notably, from the perspective of Alice, only statistical CSI is available. Therefore, the average detection error probability is derived as the covertness metric \cite{Fading3}.

\begin{theorem}
	The average detection error probability at Willie with optimal detection threshold under Rayleigh fading channels can be derived as
	\begin{equation}\label{ave_Pd}
		\begin{aligned}
			&\overline {{\xi }}\left( {{\tau ^*}} \right)  =1 - \\
			 &\frac{\pi }{{B{P_a}{\lambda _{aw}}\Gamma \!\left(\! n \!\right)\!}}\sum\limits_{i = 1}^B {\!\left[\! {\gamma \!\left(\!\! {n,\frac{{n\left( {{\sigma ^2} \!+\! \tan {\theta _i}} \right)}}{{\tan {\theta _i}}}\ln \!\left(\! {\frac{{{\sigma ^2} \!+\! \tan {\theta _i}}}{{{\sigma ^2}}}} \!\right)\!} \!\right)\!} \right.}  \\
			&\left. { \!- \gamma \!\left(\! {n,\frac{{n{\sigma ^2}}}{{\tan {\theta _i}}}\ln \!\left(\! {\frac{{{\sigma ^2}_w \!+\! \tan {\theta _i}}}{{{\sigma ^2}}}} \!\right)\!} \!\right)\!} \right]\frac{{{e^{ \!- \frac{{\tan {\theta _i}}}{{{P_a}{\lambda _{aw}}}}}}\sqrt {{\theta _i}\left( {\frac{\pi }{2} \!-\! {\theta _i}} \right)} }}{{{{\cos }^2}{\theta _i}}},
		\end{aligned}
	\end{equation}
	where $B$ is the parameter of Gaussian-Chebyshev Quadrature, and ${\theta _i} = \frac{\pi }{4}\left( {1 + \cos \frac{{\left( {2i - 1} \right)\pi }}{{2B}}} \right)$.
\end{theorem}

\begin{proof}	
	By substituting (\ref{opt_threshold}) into (\ref{ori}) and considering the probability density function (PDF) of Rayleigh fading channel, the average detection error probability can be expressed as
	\begin{equation}\label{equ_proof1}
		\begin{aligned}
			&\overline {{\xi}}\left( {{\tau ^*}} \right)= 1 - \frac{1}{{{P_a}{\lambda _{aw}}\Gamma \left( n \right)}} \times\\
			 & \!\!\int\limits_0^{ \!+\! \infty } {\!\left[\! {\gamma \!\!\left(\! {n,\!\frac{{n\!\left(\! {{\sigma ^2} \!+\! x} \right)}}{x}\!\!\ln\!\! \left(\! {\frac{{{\sigma ^2} \!\!+\!\! x}}{{{\sigma ^2}}}} \!\!\right)\!\!} \!\right)\! \!-\! \gamma \!\!\left(\! {n\!,\!\frac{{n{\sigma ^2}}}{x}\!\ln \!\left(\!\! {\frac{{{\sigma ^2} \!\!+\!\! x}}{{{\sigma ^2}}}} \!\!\right)\!} \!\right)\!} \!\right]\!\!{e^{ \frac{-x}{{{P_a}\!{\lambda _{aw}}}}}}\!dx}.\!
		\end{aligned}
	\end{equation}

By substituting $x=\tan \theta$ into (\ref{equ_proof1}) and applying Gaussian-Chebyshev Quadrature into the above integral expression \cite{Quadrature}, (\ref{ave_Pd}) can be obtained, and the proof is completed.
\end{proof}

Due to the complicated form of (\ref{ave_Pd}), it is intractable to further guide the system design. Thus, a tractable lower approximation of the detection error probability in one transmission round is derived first, and then a lower approximation of the average detection error probability is derived.

\begin{theorem}
A lower approximation of the minimum detection error probability in one transmission round is given by
	\begin{equation}\label{app1}
		{\xi ^{l}(\tau^*)}\!\!=\!\! \begin{cases} \!1 \!-\! \frac{{{e^{ - n}}{n^n}}}{{\Gamma \left( n \right)}}\ln \!\left(\! {1 \!+\! \frac{{{P_a}{{\left| {{h_{aw}}} \right|}^2}}}{{{\sigma ^2}}}} \!\right)\!\!,\! & \!\frac{{{P_a}{{\left| {{h_{aw}}} \right|}^2}}}{{{\sigma ^2}}} \!<\!  {{e^{\frac{{\Gamma \left( n \right)}}{{{e^{ - n}}{n^n}}}}} \!-\! 1}  \\0,\! & \!\frac{{{P_a}{{\left| {{h_{aw}}} \right|}^2}}}{{{\sigma ^2}}} \!\ge\!  {{e^{\frac{{\Gamma \left( n \right)}}{{{e^{ - n}}{n^n}}}}} \!-\! 1} \end{cases}
	\end{equation}
\end{theorem}

\begin{proof}
	See Appendix \ref{proof_theorem2}.
\end{proof}

The lower approximation of the minimum detection error probability of (\ref{app1}) is tighter than the approximation based on KL divergence (i.e., ${\xi ^{KL}} = 1 - \sqrt {\frac{1}{2}\mathcal{D}\left( {{\mathbb{P}_0}||{\mathbb{P}_1}} \right)}$, see Appendix \ref{proof_corollary1} for the detailed definition), which is widely used to evaluate the covertness performance in the existing works \cite{ECT2,DF}. The detailed proof is given in Appendix \ref{proof_corollary1}.

The above concise approximation facilitates the performance analysis and optimization design for the covert communication system. It can be used as a metric for the system with AWGN channels \cite{Gaussian} or the fading channels when only considering one transmission round \cite{ECT2}. Besides, it can also be adopted to analyze the average detection error probability in fading channels as follows \cite{ECT2,Nmin}.

Based on Theorem 2, the average detection error probability at Willie is derived as follows.
\begin{equation}\label{ave_Pd_app1}
	\begin{aligned}
		&\overline \xi  \left( {{\tau ^*}} \right)    \!\approx\! \int\limits_0^{{\sigma ^2}({e^{\frac{{\Gamma (n)}}{{{e^{ - n}}{n^n}}}}} \!-\! 1)} {\frac{{{e^{ \frac{-x}{{{P_a}{\lambda _{aw}}}}}}}}{{{P_a}{\lambda _{aw}}}}\!\left(\! {1 \!-\! \frac{{{e^{ - n}}{n^n}}}{{\Gamma (n)}}\ln \left( {1 \!+\! \frac{x}{{{\sigma ^2}}}} \right)} \!\right)\!dx}  \\
		& \!=\! 1 \!-\! \frac{{{e^{ - n}}{n^n}}}{{\Gamma \left( n \right)}}{e^{\frac{{{\sigma ^2}}}{{{P_a}{\lambda _{aw}}}}}}\left[ {{{\rm{E}}_{\rm{1}}}\!\left(\! {\frac{{{\sigma ^2}}}{{{P_a}{\lambda _{aw}}}}} \!\right)\! \!-\! {{\rm{E}}_{\rm{1}}}\!\left(\! {\frac{{{\sigma ^2}}}{{{P_a}{\lambda _{aw}}}}{e^{\frac{{\Gamma \left( n \right)}}{{{e^{ - n}}{n^n}}}}}} \!\right)\!} \right].
	\end{aligned}
\end{equation}

The expression of average detection error probability (\ref{ave_Pd}) and its approximation (\ref{ave_Pd_app1}) can be extended to the covert communication scenario with a passive warder by setting $P_w=0$. Besides, the lower approximation of the detection error probability in one transmission round given in (\ref{app1}) can also be extended to the scenario with a passive warder by setting $P_w=0$, which can replace the KL divergence approximation widely used in the existing works since the proposed concise approximation is tigher than the conventional one as proved in Appendix \ref{proof_corollary1}.

\section{Reliability Performance Analysis and System Design}
In this section, to analyze the reliability performance, the decoding error probability is derived. Then, the effective throughput is maximized by jointly optimizing the transmit power, the transmission rate and the blocklength, where both the covertness and the reliability requirements are considered.

\subsection{Average Decoding Error Probability at Bob}
Since both the transmit power at Alice and the jamming power at Willie affect the received SNR at Bob, considering the fading channels, the PDF of SNR can be derived as
\begin{equation}
	\begin{aligned}
		 &{f_{{\gamma _b}}}\!(t)\!\!=\!\! \frac{{d\Pr \!\left(\! {{\gamma _b} \!\!<\! t} \right)}}{{dt}} \!\!=\!\! \frac{d}{{dt}}\Pr \!\left(\! {{P_a}{{\left| {{h_{ab}}} \right|}^2} \!\!<\! t{P_w}{{\left| {{h_{wb}}} \right|}^2} \!\!+\! t\sigma _b^2} \!\right)\!\\
		& = \frac{{\sigma _b^2({P_a}{\lambda _{ab}} + {P_w}{\lambda _{wb}}t) + {P_a}{P_w}{\lambda _{ab}}{\lambda _{wb}}}}{{{{({P_a}{\lambda _{ab}} + {P_w}{\lambda _{wb}}t)}^2}}}{e^{ - \frac{{\sigma _b^2}}{{{P_a}{\lambda _{ab}}}}t}}.
	\end{aligned}
\end{equation}

Based on the linear approximation of Q-function given in \cite{ECT2} and the PDF of SNR given above, the average decoding error probability can be derived as
\begin{equation}\label{Pe}
	\begin{aligned}
		&\overline \delta  \! \approx\!\!\!\! \int\limits_0^{\alpha  - \frac{1}{{2\beta }}}\!\!\!\! {{f_{{\gamma _b}}}\left( t \right)dt}  \!+\!\!\!\!\! \int\limits_{\alpha  - \frac{1}{{2\beta }}}^{\alpha  + \frac{1}{{2\beta }}} {\!\!\left[ {\frac{1}{2} \!-\! \beta \left( {t \!-\! \alpha } \right)} \right]} {f_{{\gamma _b}}}\left( t \right)dt\\
		&\!=\! 1 \!\!-\!\! \frac{{{P_a}{\lambda _{ab}}{e^{ \!-\! \frac{{\sigma _b^2}}{{{P_a}{\lambda _{ab}}}}\left( {\alpha  - \frac{1}{{2\beta }}} \right)}}}}{{{P_a}{\lambda _{ab}} \!+\! {P_w}{\lambda _{wb}}\!\left(\! {\alpha  \!-\! \frac{1}{{2\beta }}} \!\right)\!}} \!+\! g(\alpha  \!+\! \frac{1}{{2\beta }}) \!-\! g(\alpha  \!-\! \frac{1}{{2\beta }}),
	\end{aligned}
\end{equation}
where $\alpha  \!=\! {2^R} \!-\! 1$, $\beta  \!=\! \sqrt {\frac{n}{{2\pi \left( {{4^{R}} - 1} \right)}}} $ and $g(x) = {e^{ - \frac{{\sigma _b^2x}}{{{P_a}{\lambda _{ab}}}}}}\frac{{{P_a}{\lambda _{ab}}}}{{{P_w}{\lambda _{wb}}}}\left( {\frac{{2\beta {P_w}{\lambda _{wb}}x - 2\alpha \beta {P_w}{\lambda _{wb}} - {P_w}{\lambda _{wb}}}}{{2({P_a}{\lambda _{ab}} + {P_w}{\lambda _{wb}}x)}}} \right) + {e^{\frac{{\sigma _b^2}}{{{P_w}{\lambda _{wb}}}}}}\frac{{\beta {P_a}{\lambda _{ab}}}}{{{P_w}{\lambda _{wb}}}}{{\text{E}}_1}\left( {\frac{{x\sigma _b^2}}{{{P_a}{\lambda _{ab}}}} + \frac{{\sigma _b^2}}{{{P_w}{\lambda _{wb}}}}} \right)$.

The above result can be extended to the scenario with a passive warder by replacing $P_w=0$, and the corresponding average decoding error probability can be simplified as
\begin{equation}
	\overline \delta \approx 1 - \frac{{{P_a}{\lambda _{ab}}\beta }}{{\sigma _b^2}}{e^{ - \frac{{\sigma _b^2\alpha }}{{{P_a}{\lambda _{ab}}}}}}\left( {{e^{\frac{{\sigma _b^2}}{{2{P_a}{\lambda _{ab}}\beta }}}} - {e^{ - \frac{{\sigma _b^2}}{{2{P_a}{\lambda _{ab}}\beta }}}}} \right).
\end{equation}

\subsection{Effective Throughput Maximization Optimization}
Based on the above analysis, an optimization problem can be formulated to maximize the effective throughput of the system subject to the combined constraints of covertness, reliability, blocklength, and transmit power.
\begin{align} 
	\text{(P1)}:&\max_{P_a,R,n}  \eta\\
	\text { s.t. } &\overline \xi  \left( {{\tau ^*}} \right)\ge 1-\varepsilon, \tag{16a} \label{con1}\\
	&\overline \delta  \le \kappa , \tag{16b} \label{con2}\\
	& P_a \le P_a^{max}, \tag{16c} \label{con3}\\
	& n_{min} \le n \le n_{max}, n\in \mathbb{N}^+, \tag{16d} \label{con4}
\end{align}
where (16a) and (16b) denote the covertness and reliability constraints with predetermined covertness and reliability requirements $\varepsilon$, $\kappa$, respectively. (16c) denotes the transmit power constraints at Alice with the maximum power $P_a^{max}$. Besides, (16d) denotes the blocklength constraint due to delay requirements and channel coding requirements with the maximum blocklength $n_{max}$ and the minimum blocklength $n_{min}$.

To solve this optimization problem with coupled optimization variables $P_a$, $R$ and $n$, a joint optimization framework is proposed as follows, which which involves a two-layer process. In the inner-layer, the optimal transmit power $P_a^*$ and the optimal transmission rate $R^*$ are derived with a given blocklength $n$. In the outer-layer, the optimal blocklength $n^*$ can be obtained via an exhaustive search over $\left[ {{n_{\min }},{n_{\max }}} \right]$ where $P_a^*$ and $R^*$ are calculated with each value of $n$. Finally, the globally optimal solutions $P_a^*$, $R^*$ and $n^*$ can be obtained by the above framework. The details are elaborated below.

{\textbf{Inner-layer stage:}} When $n$ is given, the optimal transmit power can be derived by $P_a^* = \min \left\{ {P_a^{\max },P_a^o} \right\}$, where ${P_a^o}$ is the solution of $\overline \xi  \left( {{\tau ^*}} \right) = 1 - \varepsilon $. This is because  $\overline \xi  \left( {{\tau ^*}} \right)$ and $\overline \delta$ decrease with $P_a$, and the effective throughput increases with $P_a$. After the optimal transmit power $P_a^*$ is obtained, we can obtain the optimal transmission rate $R^*$ as follows. It is verified that the effective throughput is a first-increasing and then-deceasing function of $R$ \cite{jammer}. Consequently, the optimal $R^o$ that maximizes $\eta$ can be effectively calculated using the bisection method. Considering the reliability constraint (\ref{con2}), the maximum transmission rate $R^{max}$ can be derived by solving $\overline \delta = \kappa $. Thus, the optimal transmission rate ${R^*} = \min \left\{ {{R^o},{R^{\max }}} \right\}$.

{\textbf{Outer-layer stage:}} Considering that the blocklength $n$ affects the constraints (\ref{con1}), (\ref{con2}) and the objective function, it is difficult to derive the optimal solutions directly. By exhaustive search on $n$ over $\left[ {{n_{\min }},{n_{\max }}} \right]$, the globally optimal solutions for (P1) and the maximum effective throughput $\eta^*$ can be obtained.

\section{Numerical Simulation}

In this section, we provide numerical results to show the covert and reliable performance of the short-packet communication system against a proactive warder. The parameter settings are as follows, unless specified otherwise: the fading parameters $\lambda_{ab}=5\times10^{-2},\lambda_{aw}=\lambda_{wb}=10^{-3}$, the AWGN variances $\sigma_b^2=\sigma_w^2=10^{-1}$ (W), the self-interference cancellation coefficient $\phi=10^{-4}$, the blocklength $n=100$, the minimum blocklength $n_{min}=50$, the maximum blocklength $n_{max}=200$, the maximum transmit power $P_a^{max}=5 (W)$ the covertness requirement $\varepsilon=10^{-1}$ and the reliability requirement $\kappa=10^{-1}$. All the simulation results shown in this paper are obtained by averaging over $10^{6}$ channel realizations.

\begin{figure}
	\centering
	\includegraphics[width=0.85\linewidth]{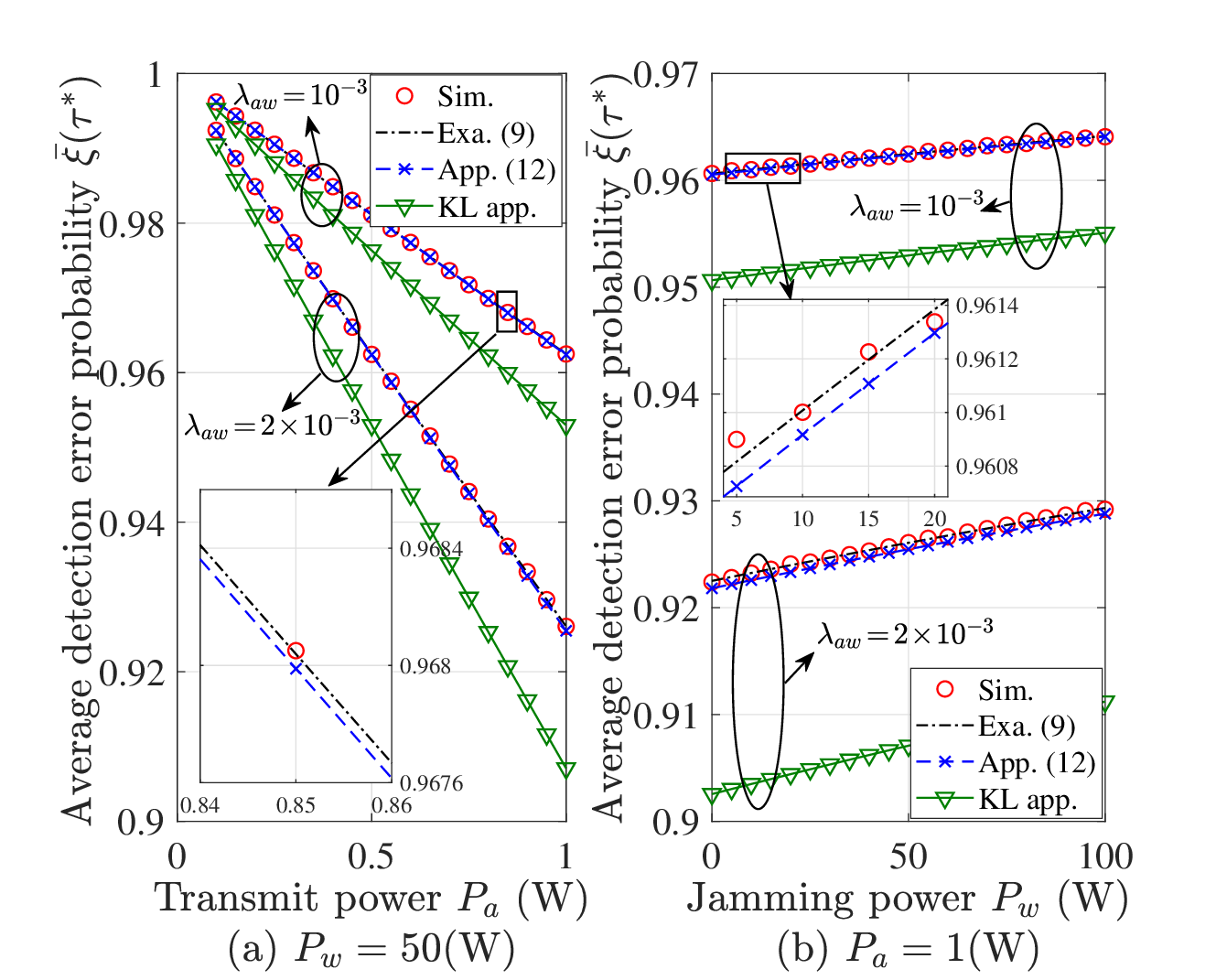}
	\caption{The detection error probability versus the transmit $/$ jamming power.}
	\label{fig_sim}
\end{figure}

In Fig.2, the impact of the transmit (jamming) power on the average detection error probability is investigated. The curves with ``Sim.'', ``Exa. (9)'', ``App. (12)'', and ``KL app.'' denote the results obtained by numerical simulations, the exact analytical expression of (\ref{ave_Pd}), the approximation (\ref{ave_Pd_app1}), and the numerical integration combined with the KL divergence of (\ref{KL_app}), respectively. It can be seen that the curves with numerical simulations, (\ref{ave_Pd}) and (\ref{ave_Pd_app1}) almost coincide. However, a significant gap exists between the curves with the KL divergence approximation and the simulation results. Besides, the average detection error probability decreases with $P_a$, and increases with $P_w$. These results validate the results given in Section III, implying that the proposed approximation expression (\ref{ave_Pd_app1}) can be adopted as a covertness performance metric to replace the widely used KL divergence metric due to its conciseness and tightness.

\begin{figure}
	\centering
	\includegraphics[width=0.85\linewidth]{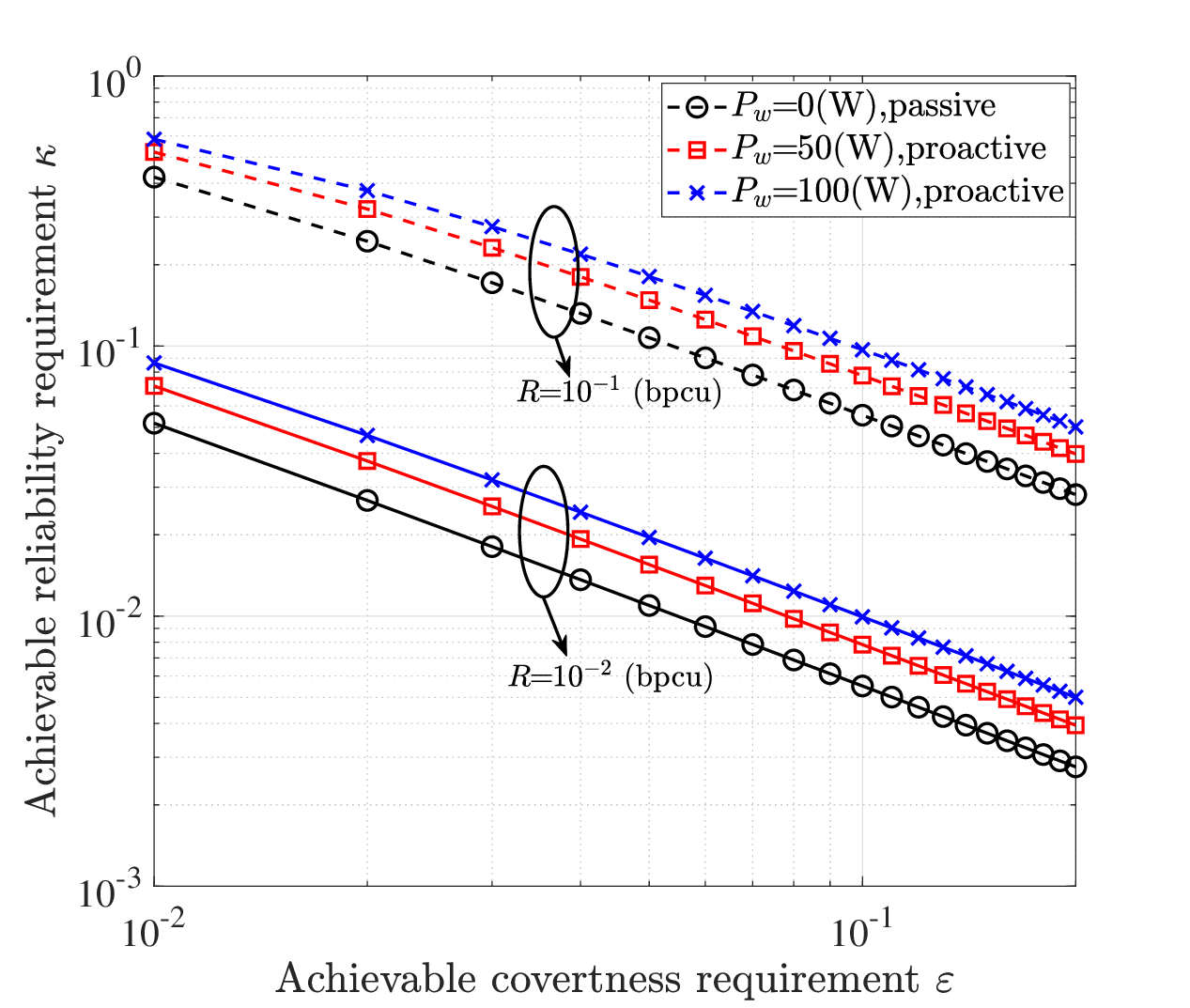}
	\caption{The achievable reliability requirement versus the achievable covertness requirement.}
	\label{fig_sim}
\end{figure}

In Fig. 3, the relationship among the achievale covertness requirement $\varepsilon$ and the achievale reliability requirement $\kappa$ is investigated. It can be seen from the figure, larger $\varepsilon$ is tolerant, and smaller $\kappa$ can be achieved. Conversely, larger $\kappa$ is tolerant, and smaller $\varepsilon$ can be achieved. Besides, it can be seen that the system performance (covertness and reliability performance) is degraded by the proactive warder $P_w=50,100$ (W) compared with the passive one $P_w=0$. These results show that the tradeoff between transmission covertness and reliability is changed by the proactive warder, and the proposed performance evalutions in Sections III and IV can be adopted to guide the system design in this case.

\begin{figure}
	\centering
	\includegraphics[width=0.85\linewidth]{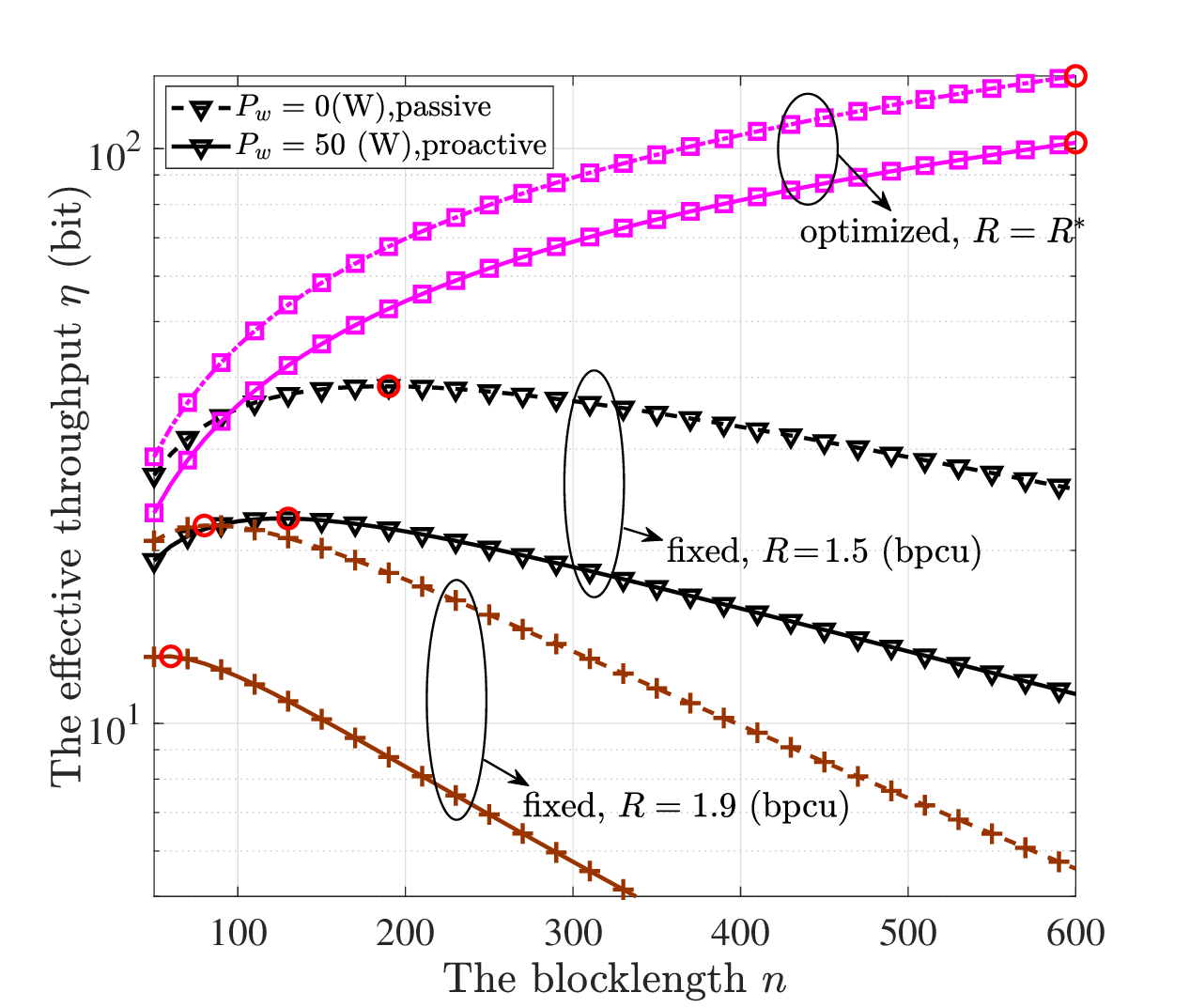}
	\caption{The effective throughput with optimized$/$fixed transmission rates versus the blocklength.}
	\label{fig_sim}
\end{figure}
In Fig.4, the impact of blocklength on the effective throughput is shown where the transmission rate is either fixed or optimized. The curves with marked solid lines and marked dotted lines denote the system performance in the system with a proactive warder and that with a passive warder, respectively. The red dots in the figure indicate the optimal blocklength that maximizes the throughput. It can be seen that in the system with fixed transmission rates, the effective throughput first increases and then decreases with $n$. This is because when $n$ is too small, $\eta$ is directly limited by the blocklength. On the contrary, when $n$ is too large, the transmit power is limited by the covertness constraint and the decoding error is too large, resulting in the reduction of effective throughput.  In addition, the effective throughput with an optimized transmission rate is always higher than that with a fixed transmission rate,  which demonstrates the feasibility of the proposed optimization framework. These results imply that for the system with optimized rates, a longer blocklength is always beneficial to improve the effective throughput. However, for the system with a fixed rate, the optimal blocklength is not necessarily the maximum one, which is critical for the system design.

\section{Conclusion}
In this paper, we investigated the reliable and covert performance of short-packet communication systems against a proactive warder. Specifically, the average detection error probability and its approximation were derived to evaluate the covertness performance. In addition, the average decoding error probability was derived to evaluate the reliability performance. Based on the analysis above, an optimization framework was proposed to maximize the effective throughput. Numerical results verified the feasibility of the proposed approximations and the optimization framework. The performance loss brought by a proactive warder was investigated compared with the passive one, and the optimal blocklength to maximize the effective throughout was elaborated with different systems.

\appendices
\section{Proof of Theorem 2}\label{proof_theorem2}
We denote ${f^u}(x) \!=\! n\left( {\frac{1}{x} \!+\! 1} \right)\ln \left( {1 \!+\! x} \right)$ , ${f^l}(x) \!=\! {\frac{n}{x}\ln \left( {1 \!+\! x} \right)}$ and $g\left( x \right) = \int_{{f^l}\left( x \right)}^{{f^u}\left( x \right)} {{e^{ - t}}{t^{n - 1}}dt} $ where $x = \frac{{{P_a}{{\left| {{h_{aw}}} \right|}^2}}}{{\sigma^2}}$. In the high covertness scenarios, $x$ always approaches zero to meet the covertness requirement, resulting in ${f^u}(x) \to n$ and ${f^l}(x) \to n$. Thus, $g\left( x \right)$ can be approximated as $	g\left( x \right)  \approx {e^{ - n}}{n^n}\left( {\frac{{{\sigma ^2} + P_a{{\left| {{h_{aw}}} \right|}^2}}}{{{\sigma ^2}}}} \right)$
and the approximation (\ref{app1}) is obtained.

Below, we prove (\ref{app1}) is a lower bound of (\ref{ori}). 

When $x \ge {e^{\frac{{\Gamma \left( n \right)}}{{{e^{ - n}}{n^n}}}}} - 1$, it holds that ${\xi ^{l}(\tau^*)} = 0 < 1 - \frac{g\left( x \right)}{{\Gamma (n)}}$.

When $0 \le x < {e^{\frac{{\Gamma (n)}}{{{e^{ - n}}{n^n}}}}} - 1$, we define the function ${f_1}(x) = g\left( x \right)  - {e^{ - n}}{n^n}\ln \left( {1 + x} \right)$ with $f_1(0) = 0$ and $\frac{{d{f_1}(x)}}{{dx}} = \frac{1}{{(x + 1)}}\left[ {{{\left( {e\frac{{\ln \left( {1 + x} \right)}}{{{{(1 + x)}^{\frac{1}{x}}}x}}} \right)}^n} - 1} \right]$.
We denote the function $f_2(x) = \frac{{\ln \left( {1 + x} \right)}}{{{{(1 + x)}^{\frac{1}{x}}}x}}$ with $f_2(0) = \frac{1}{e}$ and $\frac{{d{f_2}\left( x \right)}}{{dx}} \!=\!   \frac{{(\log (x + 1)-x)((x + 1)\log (x + 1) - x)}}{{{{\left( {x + 1} \right)}^{1/x + 1}}{x^3}}} < 0$. Thus, the first-order derivative of $f_1(x)$ is small than 0, and $f_1(x) \le f_1(0) = 0$. 

Thus, for $0 \le x < {e^{\frac{{\Gamma (n)}}{{{e^{ - n}}{n^n}}}}} - 1$, we can obtain $	{\xi ^{l}(\tau^*)} = 1 - \frac{{{e^{ - n}}{n^n}}}{{\Gamma \left( n \right)}}\ln \left( {1 + x} \right) \!<\! 1 \!-\! \frac{g\left( x \right)}{{\Gamma (n)}}$, and Theorem 2 is proved.

\section{Comparison between (11) and KL Divergence}\label{proof_corollary1}
By adopting the Pinsker’s inequality, a lower bound of minimum detection error probability is given by \cite{ECT2,DF}
\begin{equation}\label{KL_app}
	{\xi ^{KL}} = 1 - \sqrt {\frac{1}{2}\mathcal{D}\left( {{\mathbb{P}_0}||{\mathbb{P}_1}} \right)},
\end{equation}
where ${\mathbb{P}_0}$ and ${\mathbb{P}_1}$ denote the probability distributions of the observations at Willie under $\mathcal{H}_0$ and $\mathcal{H}_1$, respectively. $\mathcal{D}\left( {{\mathbb{P}_0}||{\mathbb{P}_1}} \right)$ is the KL divergence from ${\mathbb{P}_0}$ to ${\mathbb{P}_1}$ as $	\mathcal{D}\!\left( {{\mathbb{P}_0}||{\mathbb{P}_1}} \right)\! \!=\! n\!\left(\! {\ln\! \!\left(\! {1 \!+\! \frac{{{P_a}{{\left| {{h_{aw}}} \right|}^2}}}{{{\sigma ^2}}}} \!\right)\! \!+\! \frac{{{\sigma ^2}}}{{{\sigma ^2} \!+\! {P_a}{{\left| {{h_{aw}}} \right|}^2}}} \!-\! 1} \!\right)\!$.

Then, we prove that (\ref{app1}) is tighter than the KL divergence approximation, i.e., $\xi \left( {{\tau ^*}} \right) > {\xi ^{l}(\tau^*)} > {\xi ^{KL}}$.

We denote ${f_3}(x) \!\!\!=  \!\!-\! 2{e^{ - 2n}}{n^{2n - 1}}{\left( {\Gamma (n)} \right)^{ - 2}}{\ln ^2}x \!+\! \ln x \!+\! \frac{1}{x} \!\!-\!\! 1$ with $\frac{{d{f_3}(x)}}{{dx}} \!=\! \frac{1}{x}( \!-\! 2{e^{ - 2n}}{n^{2n - 1}}{\left( {\Gamma (n)} \right)^{ - 2}}\ln x \!+\! 1 \!-\! \frac{1}{x})$. In addition, we denote ${f_4}(x) =  - 2{e^{ - 2n}}{n^{2n - 1}}{\left( {\Gamma (n)} \right)^{ - 2}}\ln x + 1 - \frac{1}{x}$ with $\frac{{d{f_4}(x)}}{{dx}} = {x^{ - 2}} - 2{e^{ - 2n}}{n^{2n - 1}}{\left( {\Gamma (n)} \right)^{ - 2}}{x^{ - 1}}$, where $2{e^{ - 2n}}{n^{2n - 1}}{\left( {\Gamma (n)} \right)^{ - 2}} < 1$, proved as follows.

By denoting $M_1(n) = \frac{{\Gamma (n)}}{{\sqrt {2n} {e^{ - n}}{n^{n - 1}}}}$, we can obtain $\frac{{\partial M_1(n)}}{{\partial n}} =  - \frac{{\sqrt 2 }}{4}{e^n}{n^{ - n - 3/2}}n!\left( {2n\log (n) - 2n{\psi ^{(0)}}(n) - 1} \right) < 0$, and $M_1(1) = \frac{e}{{\sqrt 2 }} > 1$, $\mathop {\lim }\limits_{n \to \infty } M_1(n) = \sqrt \pi   + \mathcal{O}\left( {\frac{1}{n}} \right) > 1$. Thus, $2{e^{ - 2n}}{n^{2n - 1}}{\left( {\Gamma (n)} \right)^{ - 2}} \in \left( {\frac{1}{\pi },\frac{2}{{{e^2}}}} \right)$.

Therefore, $f_4(x)$ increases in $\left[ {1,\frac{1}{2}{e^{2n}}{n^{1 - 2n}}{{\left( {\Gamma (n)} \right)}^2}} \right)$ and decreases in $\left( {\frac{1}{2}{e^{2n}}{n^{1 - 2n}}{{\left( {\Gamma (n)} \right)}^2}, + \infty } \right)$. In addition, $f_4(1) = 0$ and $f_4(x)$ is larger than 0 in the interval $\left[ {1,{\kappa _1}} \right)$ and less than 0 in the interval $\left( {{\kappa _1}, + \infty } \right)$, where ${\kappa _1}$ is the solution to $ \!-\! 2{e^{ \!-\! 2n}}{n^{2n - 1}}{\left( {\Gamma (n)} \right)^{ - 2}}x\ln x \!+\! x \!\!\!=\!\!\! 1$ except 1. Futhermore, $f_3(x)$ increases in $\left[ {1,{\kappa _1}} \right)$ and decreases in $\left( {{\kappa _1}, + \infty } \right)$. When $x = {e^{\frac{{\Gamma (n)}}{{{e^{ - n}}{n^n}}}}}$, we can obtain $	{f_3}({e^{\frac{{\Gamma (n)}}{{{e^{ - n}}{n^n}}}}}) = \frac{1}{n}\left( {n\left( {\frac{{\Gamma (n)}}{{{e^{ - n}}{n^n}}} + {e^{ - \frac{{\Gamma (n)}}{{{e^{ - n}}{n^n}}}}} - 1} \right) - 2} \right)$, where the sequence ${M_2}(n) = n\left( {\frac{{\Gamma (n)}}{{{e^{ - n}}{n^n}}} + {e^{ - \frac{{\Gamma (n)}}{{{e^{ - n}}{n^n}}}}} - 1} \right)$ is monotonically increasing with $n$ and $M_2(1)=1.78<2$, $M_2(2)=2.01>2$. Therefore, $f_3(x)>0$ in $\left[ {1,{e^{\frac{{\Gamma (n)}}{{{e^{ - n}}{n^n}}}}}} \right]$ with $n \ge 2$. And by substituting $x = 1 + \frac{{P_a{{\left| {{h_{aw}}} \right|}^2}}}{{{\sigma ^2}}}$ into $f_3(x)$, and adopting the results derived above, we can obtain ${\xi ^{l}(\tau^*)}\!=\! 1 \!-\! \frac{{{e^{ - n}}{n^n}}}{{\Gamma (n)}}\ln \!\left(\! {1 \!+\! \frac{{{P_a}{{\left| {{h_{aw}}} \right|}^2}}}{{{\sigma ^2}}}} \!\right)\! \!>\! 1 \!-\! \sqrt {\frac{1}{2}\mathcal{D}\left( {{\mathbb{P}_0}||{\mathbb{P}_1}} \right)}$ for $0 < \frac{{{P_a}{{\left| {{h_{aw}}} \right|}^2}}}{{{\sigma ^2}}} \le \left( {{e^{\frac{{\Gamma \left( n \right)}}{{{e^{ - n}}{n^n}}}}} - 1} \right)$.

Besides, when $\frac{{{P_a}{{\left| {{h_{aw}}} \right|}^2}}}{{{\sigma ^2}}} > {e^{\frac{{\Gamma \left( n \right)}}{{{e^{ - n}}{n^n}}}}} - 1,{\xi ^{l}(\tau^*)} = 0 > {\xi ^{KL}}$. Note that the above results hold with the assumption $n \ge 2$, which is always true in short-packet communications.

\bibliographystyle{IEEEtran}         
\bibliography{reference}   
\vspace{12pt}
\end{document}